%% file: main.tex
\documentclass{article}

\usepackage{epsfig,color}

\usepackage{geometry}
\usepackage{url}
\newcommand\states[1]{\Sigma_{#1}}
\newcommand\Trans[4]{{#2} \stackrel{#3}{\longrightarrow_{#1}}{#4}}

\newcommand\tmlcga{\textnormal{TML-CGA}\xspace}
\newcommand\noreccga{\textnormal{NORec-CGA}\xspace}
\newcommand\ronoreccga{\textnormal{NORec2-CGA}\xspace}
\newcommand\tml{\textnormal{TML}\xspace}
\newcommand\norec{\textnormal{NORec}\xspace}
\newcommand\ronorec{\textnormal{NORec2}\xspace}

\newcommand{\complete}[1]{[#1]}

\newcommand{\imp}{\Rightarrow}

\newcommand{\undef}{\bot}

\newcommand{\pprec}{\prec\!\!\!\prec}

\usepackage{algorithm}
\usepackage{algorithmicx}
\usepackage[noend]{algpseudocode}

\floatname{algorithm}{Listing}
\def\lseq{\langle}
\def\rseq{\rangle}
\def\eseq{\varepsilon}
\usepackage{xspace}

\newcommand{\OMIT}[1]{}

\newcommand{\bftmbegin}{\texttt{\bf TXBegin}\xspace}
\newcommand{\bftmwrite}{\texttt{\bf TXWrite}\xspace}
\newcommand{\bftmread}{\texttt{\bf TXRead}\xspace}
\newcommand{\bftmcommit}{\texttt{\bf TXCommit}\xspace}

\newcommand{\tmbegin}{\textnormal{\texttt{TXBegin}}\xspace}
\newcommand{\tmcommit}{\textnormal{\texttt{TXCommit}}}
\newcommand{\tmread}{\textnormal{\texttt{TXRead}}\xspace}
\newcommand{\tmabort}{\textnormal{\texttt{TXAbort}}}
\newcommand{\tmwrite}{\textnormal{\texttt{TXWrite}}\xspace}

\newcommand{\abegin}{\textnormal{\texttt{ATXBegin}}\xspace}
\newcommand{\acommit}{\textnormal{\texttt{ATXCommit}}\xspace}
\newcommand{\aread}{\textnormal{\texttt{ATXRead}}\xspace}
\newcommand{\awrite}{\textnormal{\texttt{ATXWrite}}\xspace}

\newcommand{\Validate}{\textnormal{\texttt{Validate}}}

\newcommand{\dom}{\textit{dom}}
\newcommand{\sdef}{\mathrel{\widehat{=}}}

\newcommand{\sif}{\textbf{if}}
\newcommand{\sthen}{\textbf{then}}
\newcommand{\selse}{\textbf{else}}

\newcommand{\reffig}[1]{Figure~\ref{#1}}
\newcommand{\refthm}[1]{Theorem~\ref{#1}}
\newcommand{\reflem}[1]{Lem\-ma~\ref{#1}}

\newcommand{\refsec}[1]{Section~\ref{#1}}
\newcommand{\refex}[1]{Example~\ref{#1}}
\newcommand{\refdef}[1]{Definition~\ref{#1}}
\newcommand{\reflst}[1]{Listing~\ref{#1}}

\newcommand{\beginInv}[1]{\ensuremath{inv_{#1}(\tmbegin)}}
\newcommand{\beginResp}[1]{\ensuremath{resp_{#1}(\tmbegin)}}
\newcommand{\readInv}[2]{\ensuremath{inv_{#1}(\tmread(#2))}}
\newcommand{\readResp}[2]{\ensuremath{resp_{#1}(\tmread(#2))}}
\newcommand{\writeInv}[2]{\ensuremath{inv_{#1}(\tmwrite(#2))}}
\newcommand{\writeResp}[1]{\ensuremath{resp_{#1}(\tmwrite)}}
\newcommand{\commitInv}[1]{\ensuremath{inv_{#1}(\tmcommit)}}
\newcommand{\commitResp}[1]{\ensuremath{resp_{#1}(\tmcommit)}}

\newcommand{\abortResp}[1]{\ensuremath{resp_{#1}(\tmabort)}}

\newcommand{\doCommitReadOnlyKW}[1]{\ensuremath{\texttt{DoCommitRO}_{#1}}}
\newcommand{\doCommitReadOnly}[2]{\ensuremath{\doCommitReadOnlyKW{#1}(#2)}}
\newcommand{\doCommitWriter}[1]{\ensuremath{\texttt{DoCommitW}_{#1}}}
\newcommand{\doRead}[2]{\ensuremath{\texttt{DoRead}_{#1}(#2)}}
\newcommand{\doWrite}[2]{\ensuremath{\texttt{DoWrite}_{#1}(#2)}}

\newcommand{\vidx}[2]{\ensuremath{\mathit{validIdx}_{#1}(#2)}}

\newcommand{\action}[3]{\ensuremath{
\begin{array}[t]{@{}l@{~}l@{}}
\multicolumn{2}{@{}l@{}}{#1}\\
\textsf{Pre:}&#2\\
\textsf{Eff:}&#3
\end{array}
}}

\newcommand{\statusText}[1]{\text{#1}}

\newcommand{\pcNotStarted}{\statusText{notStarted}}
\newcommand{\pcBeginPending}{\statusText{beginPending}}

\newcommand{\pcReady}{\statusText{ready}}
\newcommand{\pcDoWrite}{\statusText{doWrite}}
\newcommand{\pcWriteResp}{\statusText{writeResp}}
\newcommand{\pcDoRead}{\statusText{doRead}}
\newcommand{\pcReadResp}{\statusText{readResp}}
\newcommand{\pcDoCommit}{\statusText{doCommit}}
\newcommand{\pcCommitResp}{\statusText{commitResp}}

\newcommand{\pcCommitted}{\statusText{committed}}
\newcommand{\pcAborted}{\statusText{aborted}}

\newcommand{\ret}[1]{\overline{#1}}

\usepackage{amssymb,amsfonts,amsthm,amsmath}

\newcounter{thm}
\newtheorem{theorem}[thm]{Theorem}

\newtheorem{lemma}[thm]{Lemma}

\newtheorem{definition}[thm]{Definition}
\newtheorem{example}[thm]{Example}

\algblock{Atomic}{EndAtomic}
\algnotext{EndAtomic}
\algrenewtext{Atomic}{\textbf{atomic}}

\algnewcommand\algorithmicawait{\textbf{await}}
\algnewcommand\Await[1]{\State\algorithmicawait\ #1}

\algrenewcommand\algorithmicindent{1.0em}
\usepackage{varwidth}
\newlength\myindent
\begin{document}





\title{Reducing Opacity to Linearizability: \\ A Sound and Complete
  Method} 

\author{Alasdair Armstrong\qquad\qquad Brijesh Dongol \\
{\small Brunel University London, UK}
\and Simon Doherty \\{\small University of Sheffield, UK}}

\maketitle

\begin{abstract}
  Transactional memory is a mechanism that manages thread
  synchronisation on behalf of a programmer so that blocks of code
  execute with an illusion of atomicity. The main safety criterion for
  transactional memory is opacity, which defines conditions for
  serialising concurrent transactions.

  Proving opacity is complicated because it allows concurrent transactions
  to observe distinct memory states, while TM implementations are typically
  based on one single shared store. 
  This paper presents a sound and complete method, based on
  coarse-grained abstraction, for reducing proofs of opacity to the
  relatively simpler correctness condition: linearizability. We use
  our methods to verify TML and NORec from the literature and show our
  techniques extend to relaxed memory models by showing that both are
  opaque under TSO without requiring additional fences. Our methods
  also elucidate TM designs at higher level of abstraction; as an
  application, we develop a variation of NORec with fast-path reads
  transactions. All our proofs have been mechanised, either in the
  Isabelle theorem prover or the PAT model checker.
\end{abstract}



\section{Introduction}

Transactional Memory (TM) provides programmers with an easy-to-use
synchronisation mechanism for concurrent access to shared data. The
basic mechanism is a programming construct that allows one to specify
blocks of code as {\em transactions}, with properties akin to database
transactions (atomicity, consistency and isolation)
\cite{DBLP:series/synthesis/2010Harris}. Like
database transactions, a software transaction might encounter
interference and abort.  Transactions must be invisible to all other
transactions until they successfully commit.

Over the last few years, there has been an explosion of research on
TM, leading to TM libraries implemented in many programming language
libraries (Java, Clojure, C++11), compiler support for TM (G++ 4.7)
and hardware support (Intel's Haswell processor). This widespread
adoption together with their underlying complexity means that formal
verification of TM is an important problem.

The main safety condition for TM is \emph{opacity}
\cite{GuerraouiK08,2010Guerraoui}, which provides conditions for
serialising (concurrent) transactions into a sequential order and
describes the meaning of this sequential order. Over the years,
several methods for verifying opacity have been developed
\cite{Lesani14,LesaniP14,Lesani2012,GuerraouiHS10,GHS08,ASS16,DDSTW15}. A
difficulty in opacity verification is that it must deal with
sequences of memory snapshots that reflect the history of all committed
transactions. For example, a read-only transaction is allowed to serialise against
an earlier snapshot that is different from the current memory. 

This paper develops a method for simplifying proofs of opacity by
reducing it to the well-known correctness condition
\emph{linearizability} \cite{HeWi90}. Unlike opacity, linearizability
proofs need only concern themselves with a single ``current'' value
of the abstract object. This method is
\emph{sound} (any verification using linearizability implies opacity
of the original algorithm) and we show that it is also \emph{complete}
(for any opaque implementation, it is possible to prove opacity
by proving linearizability with respect to an appropriate abstract object).
In addition to reducing the complexity of the
problem, our approach makes it possible to leverage the rich literature on
linearizability verification \cite{DongolD15} to verify
opacity.

Our method involves development of a coarse-grained abstraction of the
TM implementation at hand, where the fine-grained operations of an
implementation are abstracted by coarse-grained specification with
atomic operations. The first step is to show that this coarse-grained
specification is opaque, and the second to show that the
implementation does indeed linearize to the coarse-grained
abstraction. We then leverage a soundness result to conclude that the
original implementation is opaque.

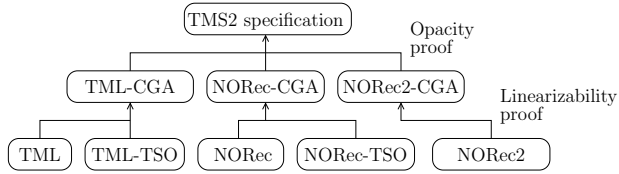
\begin{figure}[t]
  \centering
    \scalebox{0.6}{\input{overview.pspdftex}}
  \caption{Overview of proofs}
  \label{fig:overview}
\end{figure}
We use a Transactional Mutex Lock (TML) \cite{DalessandroDSSS10} as a
running example to introduce the different concepts, then apply our
method to this algorithm, i.e., verify it correct by proving
linearizability against a coarse-grained abstraction. Then to show the
applicability of our approach, we verify the more complex NORec
algorithm~\cite{DBLP:conf/ppopp/DalessandroSS10}. In addition, we show
that our method scales to relaxed-memory models; we prove opacity of
both TML and NORec under TSO memory via a proof of
linearizability. 

Proof via coarse-grained abstraction has advantages, e.g., it can be
used to (more easily) elucidate underlying design principles behind
each algorithm. In particular, we prove distinctness between TML and
NORec by generating counter-example histories at the level of their
coarse-grained abstractions. Such high-level abstractions of each
algorithm elucidate design alternatives. In particular, we develop
\ronorec, a variation of NORec, with read-only transactions that
performs fast-path validation when reading previously loaded
values. We show that this variation is opaque, and that its design is
distinct from the original NORec algorithm (at the coarse-grained
level). We develop a fine-grained version of this variation and show
that it is a linearizable with respect to this coarse-grained
abstraction. We implement this new algorithm in C and evaluate its
performance using the STAMP benchmarking suite \cite{stamp}; our
experimental results are discussed in \refsec{sec:tm-design}.

An overview of the proofs we have performed is given in
\reffig{fig:overview}, where \tmlcga is the coarse-grained abstraction
corresponding to TML. Each proof of opacity has been mechanised in the
interactive theorem prover Isabelle \cite{NipkowPW02}. A benefit of
these proofs is that they have been proved correct for models of
arbitrary size. The linearizability proofs could also have been fully
verified in Isabelle. However, we chose to exploit one of the many
methods for model checking linearizability from the literature. In
particular, we use the PAT model checker \cite{SunLDP09}, which
provides a process-algebraic encoding of the systems under
consideration and allows refinement to be checked automatically. A
benefit of the proofs in PAT is that no invariants or simulation
relations need to be defined. The method is also able to cope with
TSO-encodings of both the TML and NORec algorithms. A surprising
result in these proofs is that each implementation in
\reffig{fig:overview} is \emph{equivalent} to its coarse-grained
counterpart, i.e., for example we can show that both TML refines
TML-CGA \emph{and} TML-CGA refines TML.

This paper is structured as follows. Background material including the
TML algorithm and a formal definition of opacity is given in
\refsec{sec:background}. Our method of reducing opacity to
linearizability, including the soundness and completeness proofs are
given in \refsec{sec:reducing-opacity}, which we apply to our TML
algorithm in \refsec{sec:proving-opacity-via} (including proofs of
opacity and linearizability). The NORec algorithm and its proof is
given in \refsec{sec:norec-algorithm}. We give extensions to include
relaxed memory in \refsec{sec:relaxed-memory} and in
\refsec{sec:tm-design}, we develop \ronorec, demonstrating how
reduction via coarse-grained abstraction contributes to TM design.

\section{Opacity}
\label{sec:background}

This section formalises \emph{opacity} as defined by Guerraoui and
Kapalka \cite{2010Guerraoui}. Our formalisation mainly follows Attiya
\emph{et al.} \cite{AGHR13}, but we explicitly include the
prefix-closure constraint to ensure consistency with other accepted
definitions
\cite{LLM12,2010Guerraoui,DBLP:series/synthesis/2010Harris}. We
introduce TM using the TML algorithm in
\refsec{sec:exampl-trans-mutex}, formalise histories in
\refsec{sec:histories} and opacity in
\refsec{sec:opacity}. 

\subsection{Example: Transactional Mutex Lock}

\label{sec:exampl-trans-mutex}


\newcommand{\vloc}[1]{\mathit{loc}_{#1}}
\newcommand{\vglb}{\mathit{glb}}
\newcommand{\vmem}{\mathit{mem}}

\begin{algorithm}[t]
  \caption{The Transactional Mutex Lock (TML) algorithm}
  \label{alg:tml}

  \begin{varwidth}[t]{0.45\textwidth}
  \begin{algorithmic}[1]
    \Procedure{Init}{}
    \State $\vglb$ $\gets$ 0
    \EndProcedure
    \algstore{tml-ctr}
  \end{algorithmic}
    \begin{algorithmic}[1]
      \algrestore{tml-ctr}
      \Procedure{$\tmbegin_t$}{}
      \State {\bf do}\ $\vloc{t}$ $\gets$ $\vglb$
      \State {\bf until} {$\mathit{even}(\vloc{t})$}
      \EndProcedure\smallskip

      \Procedure{$\tmcommit_t$}{}
      \If{$\mathit{odd}(\vloc{t})$}
      \State $\vglb$ $\gets$ $\vloc{t} + 1$ \label{tml:8}
      \EndIf
      \EndProcedure
      \algstore{tml-ctr}
    \end{algorithmic}
  \end{varwidth}
  \hfill
  \begin{varwidth}[t]{0.45\textwidth}
  \begin{algorithmic}[1]
    \algrestore{tml-ctr}
    \Procedure{$\tmread_t$}{$a$}
    \State $v_t$ $\gets$ $\vmem(a)$
    \If{$\vglb = \vloc{t}$}
    \State  \Return $v_t$
    \Else\ \textbf{abort}
    \EndIf
    \EndProcedure\smallskip

    \Procedure{$\tmwrite_t$}{$a, v$}
    \If {$\mathit{even}(\vloc{t})$}
    \State {\bf if}\ {$!\mathit{cas}(\&\vglb, \vloc{t}, \vloc{t}\!+\!1)$}
    \State {\bf then}\ \textbf{abort}
    \setlength{\itemindent}{\myindent}
    \State {\bf else}\ $\vloc{t}$++
    \EndIf
    \State $\vmem(a)$ $\gets$ $v$
    \EndProcedure
  \end{algorithmic}
\end{varwidth}

\end{algorithm}

To support the concept of transactions, TM usually provide a number of
operations to programmers: operations to start (\tmbegin) or to end a
transaction (\tmcommit), and operations to read or write shared data
(\tmread, \tmwrite)\footnote{In general, arbitrary operations can be
  used here; for simplicity we use reads and writes to variables.}.
These operations can be called (invoked) from within a program
(possibly with some arguments, e.g., the variable to be read) and then
will return with a response. Except for operations that start
transactions, all other operations might potentially respond with
\tmabort, thereby aborting the whole transaction. 

We will use the Transactional Mutex Lock (TML) by Dalessandro \emph{et
  al.} \cite{DalessandroDSSS10} in \reflst{alg:tml} as a running
example.  It provides the four types of operations, but operation
\tmcommit\ in this algorithm will never respond with abort.  TML
adopts a very strict policy for synchronisation among transactions: as
soon as one transaction has successfully written to a variable, all
other transactions running concurrently will be aborted when they
execute a \tmread or \tmwrite operation. For synchronisation, TML uses
a global counter $\vglb$ (initially $0$), and each transaction $t$
uses a local variable $\vloc{t}$ to store a local copy of
$\vglb$. Variable $\vglb$ records whether there is a {\em live writing
  transaction}, i.e., a transaction which has started, has
neither committed nor aborted, and has executed (or is executing) a
write operation. More precisely, $\vglb$ is odd if there is a live
writing transaction, and even otherwise. Initially, we have no live
writing transaction and thus $\vglb$ is $0$ (and hence even).

Operation $\tmbegin_t$ copies the value of $\vglb$ into its local
variable $\vloc{t}$ and checks whether $\vglb$ is even. If so, the
transaction is started, and otherwise, the transaction attempts to
start again by rereading $\vglb$. Operation $\tmread_t$ succeeds as
long as $\vglb$ equals $\vloc{t}$ (meaning no writes have occurred
since the transaction $t$ began), otherwise it aborts the current
transaction. The first execution of $\tmwrite_t$ attempts to increment
$\vglb$ using a $\mathit{cas}$ (compare-and-swap), which atomically
compares the first and second parameters, and sets the first parameter
to the third if the comparison succeeds. If the $\mathit{cas}$ attempt
fails, a write by another transaction must have occurred, and hence,
the current transaction aborts. Otherwise $\vloc{t}$ is incremented
(making its value odd) and the write is performed. Note that because
$\vloc{t}$ becomes odd after the first successful write, all
successive writes that are part of the same transaction will perform
the write directly after testing $\vloc{t}$ at line $1$. Further note
that if the $\mathit{cas}$ succeeds, $\vglb$ becomes odd, which
prevents other transactions from starting, and causes all concurrent
live transactions still wanting to read or write to abort. Thus a
writing transaction that successfully updates $\vglb$ effectively
locks shared memory. Operation $\tmcommit_t$ checks to see if $t$ has
written to memory by testing whether $\vloc{t}$ is odd. If the test
succeeds, $\vglb$ is set to $\vloc{t}+1$. At line \ref{tml:8},
$\vloc{t}$ is guaranteed to be equal to $\vglb$, and therefore this
update has the effect of incrementing $\vglb$ to an even value,
allowing other transactions to begin.


\subsection{Histories}
\label{sec:histories}
\begin{table}[t]
\begin{center}
  \begin{tabular}{l@{\qquad\qquad}l}
    invocations~ & possible matching responses \\
    \hline \\[-6pt]
    $\tmbegin_p$ & $\ret{\tmbegin}_p$ \\[1pt]
    $\tmcommit_p$ & $\ret{\tmcommit}_p$, $\ret{\tmabort}_p$ \\[1pt]
    $\tmread_p(a)$ & $\ret{\tmread}_p(v)$, $\ret{\tmabort}_p$ \\[1pt]
    $\tmwrite_p(a,v)$ & $\ret{\tmwrite}_p$, $\ret{\tmabort}_p$
  \end{tabular}
\end{center}
Assume  $p$ is a transaction identifier
from a set of transactions $T$, $a$ is an address,  and $v$ a
value in the address.
\caption{Events of a TM implementation}
\label{fig:stmevents}
\end{table}

As standard in the literature, opacity is defined on the
\emph{histories} of an implementation.  Histories are sequences of
{\em events} that record all interactions between the TM
implementation and its clients.  
Thus each event is either an invocation or a response of a TM
operation. For the implementations we consider in this paper, possible
invocation and matching response events are given in
Table~\ref{fig:stmevents}.

We introduce notation $\bftmbegin_p$ to denote the two-element
sequence $\lseq \tmbegin_p, \ret{\tmbegin}_p \rseq$.  Similarly,
$\bftmwrite_p(x, v)$ denotes the sequence
$\lseq \tmwrite_p(x, v), \ret{\tmwrite}_p \rseq$ and
$\bftmread_p(x, v)$ denotes
$\lseq \tmread_p(x), \ret{\tmread}_p(v) \rseq$. We use notation
`$\cdot$' for sequence concatenation.
\begin{example}
\label{ex:1}
The following history is a possible execution of the TML, where the
address $x$ (initially $0$) is accessed by two transactions 2 and 3
running concurrently. \smallskip

\noindent
\hfill$
\begin{array}[b]{@{}l@{}}

  \lseq \begin{array}[t]{@{}l@{}}
          \tmbegin_3, \tmbegin_2,\ret{\tmbegin}_3 ,
          \ret{\tmbegin}_2, \tmwrite_3(x,4) \rseq \cdot {} \\
          \bftmread_2(x, 0) \cdot \lseq \ret{\tmwrite}_3 \rseq \cdot \bftmcommit_3
  \end{array}
\end{array}
$\hfill{}

\end{example}

We use the following notation on histories: for a history $h$, $h| p$
is the projection onto the events of transaction $p$ and $h[i..j]$ the
sub-sequence of $h$ from $h(i)$ to $h(j)$ inclusive.  We are interested
in three different types of histories. At the concrete level the TML
implementation produces histories of interleaved events (e.g., $h$ in
\refex{ex:1}). At an intermediate level, we are interested in
\emph{alternating histories}, where operations appear atomic, but
interleaving may occur at the level of transactions. At the abstract
level, we are interested in {\em sequential histories}, where there is
no interleaving at any level---transactions are atomic: completed
transactions end before the next transaction
starts.  

Let $\eseq$ denote the empty sequence. A history $h$ is
\emph{alternating} if $h = \eseq$ or $h$ is an alternating sequence of
invocation and matching response events starting with an invocation
and possibly ending with an invocation. We will assume that $h|p$ is
alternating for any history $h$ and transaction~$p$. Note that this
does not necessarily mean $h$ is alternating itself.  Opacity is
defined for well-formed histories, which formalise the allowable
interaction between a TM implementation and its clients. A projection
$h|p$ of a history $h$ onto a transaction $p$ is \emph{well-formed}
iff it is $\eseq$ or a sequence $s_0 \ldots s_m$ where
$s_0 = \tmbegin_p$ and for all $0 < i < m$,
$s_i \notin \{\ret{\tmbegin}_p, \ret{\tmcommit}_p,
\ret{\tmabort}_p\}$.
Furthermore, $h|p$ is \emph{committed} whenever
$s_m = \ret{\tmcommit}_p$ and \emph{aborted} whenever
$s_m = \ret{\tmabort}_p$.  In these cases, the transaction $h|p$ is
\emph{finished}, otherwise it is \emph{live}. A history is
\emph{well-formed} iff $h|p$ is well-formed for every transaction $p$.

\begin{example}
  The history in \refex{ex:1} is well-formed, and contains a committed
  transaction 3 and a live transaction 2.  $\hfill\Box$
\end{example}

\subsection{Opacity}
\label{sec:opacity}

The basic principle behind the definition of opacity (and similar
definitions) is the comparison of a given concurrent history against a
sequential one.  Within the concurrent history in question, we
distinguish between \emph{live}, \emph{committed} and \emph{aborted}
transactions.  Opacity imposes a number of constraints, that can be
categorised into three main types: (1) \emph{ordering constraints},
which describe how events occurring in a concurrent history may be
sequentialised; (2) \emph{semantic constraints} that describe validity
of a sequential history $hs$; and (3) a \emph{prefix-closure
  constraint}, which requires that each prefix of a concurrent history
can be sequentialised so that the ordering and semantic constraints
above are satisfied.
To help formalise these opacity constraints we introduce the following
notation.  We say a history $h$ is {\em equivalent} to a history $h'$,
denoted $h \equiv h'$, iff $h| p =h'| p$ for all transactions
$p \in T$.  Further, the {\em real-time order} on transactions $p$ and
$q$ in a history $h$ is defined as $p \prec_h q$ if $p$ is a completed
transaction and the last event of $p$ in $h$ occurs before the first
event of $q$. 

\paragraph{Sequential history semantics.}
We now formalise the notion of sequentiality for transactions: a
sequential history is alternating and does not interleave events of
different transactions. We first define non-interleaved histories,
noting that the definition must also cover live transactions.

\begin{definition}[Non-interleaved history]
  A well-formed history $h$ is {\em non-interleaved} if transactions
  do not overlap, i.e., for any transactions $p$ and $q$ and histories
  $h_1, h_2$ and $h_3$, if \smallskip

  $h = h_1 \cdot \lseq \tmbegin_p \rseq \cdot
  h_2 \cdot \lseq \tmbegin_q \rseq \cdot h_3$ \smallskip

  \noindent and $h_2$ contains no \tmbegin\ invocation events, then either $h_2$
  contains a response event $e$ such that
  $e \in \{\ret{\tmabort}_p, \ret{\tmcommit}_p\}$, or $h_3$ contains
  no events for a transaction $p$.
\end{definition}

\noindent In addition to being non-interleaved, a sequential history
has to ensure that the behaviour is meaningful with respect to the
reads and writes of the transactions. For this, we look at each
address in isolation and define the notion of a valid sequential
behaviour on a single address. To this end, we model shared memory by
a set $A$ of addresses mapped to values denoted by a set $V$. Hence
the type $A \rightarrow V$ describes the possible states of the shared
memory.

\begin{definition}[Valid history]\label{def:valid-hist}
  Let $h = \lseq ev_0, \ldots ,
  ev_{2n-1} \rseq$ be a sequence of alternating invocation and response
  events starting with an invocation and
  ending with a response.

  We say $h$ is {\em valid} if there exists a sequence of states
  $\sigma_0, \ldots , \sigma_{n}$ such that $\sigma_0(a) = 0$ for all
  addresses $a$, and for all $i$ such that $0 \leq i < n$ and
  $p \in T$:
  \begin{enumerate}
  \item if $ ev_{2i} = \tmwrite_p(a,v)$ and
    $ev_{2i+1} = \ret{\tmwrite}_p$
    then $\sigma_{i+1} = \sigma_{i}[a:=v]$; and
  \item if $ev_{2i} = \tmread_p(a)$ and
    $ev_{2i + 1} = \ret{\tmread}_p(v)$ then both $\sigma_i(a)=v$ and
    $\sigma_{i+1} = \sigma_{i}$ hold; and
  \item for all other pairs of events 
    $\sigma_{i+1} = \sigma_{i}$.
  \end{enumerate}
\end{definition}

The effect of the writes in a transaction must only be visible if a
transaction commits. All other writes must not affect memory.
However, all reads must be consistent with previously committed
writes. Therefore, only some histories of an object reflect ones that
could be produced by an STM. We call these the {\em legal} histories,
and they are defined as follows.

\begin{definition}[Legal history]
  \label{legal}
  Let $hs$ be a non-interleaved history, $i$ an index of $hs$, and
  $hs'$ be the projection of $hs[0..(i-1)]$ onto all events of
  committed transactions plus the events of the transaction to which
  $hs(i)$ belongs.  We say \emph{$hs$ is legal at $i$} whenever $hs'$
  is valid.  We say \emph{$hs$ is legal} iff it is legal at each index
  $i$.
\end{definition}

\noindent This allows us to define sequentiality for a single history,
which we lift to the level of specifications.

\begin{definition}[Sequential history]
  \label{def:sequential}
  A well-formed history $hs$ is {\em sequential} if it is
  non-interleaved and legal. 
  We denote by ${\cal S}$  the
  set of all possible well-formed sequential histories.
\end{definition}

\paragraph{Opaque histories.}

A given concrete history may be incomplete, i.e., consist of pending
operation calls. As some of these pending calls may have taken effect,
pending operation calls may be completed by adding matching
responses. Thus we define a function $extend$ that adds matching
responses to some of the pending invocations to a history $h$. There may
also be incomplete operation calls that have not taken effect; it is
safe to remove the pending invocations. We let $\complete{h}$ denote
the history $h$ with all pending invocations removed.
\begin{definition}[Opaque history]
  \label{def:opaque}
  \label{opaquedef}
  A history $h$ is {\em end-to-end opaque} iff for some
  $he \in extend(h)$, there exists a sequential history
  $hs \in {\cal S}$ such that $[he] \equiv hs$ and
  $\prec_{[he]} \subseteq \prec_{hs}$.  A history $h$ is {\em opaque}
  iff each prefix $h'$ of $h$ is end-to-end opaque; a set of histories
  ${\cal H}$ is {\em opaque} iff each $h \in {\cal H}$ is opaque; and
  an STM implementation is \emph{opaque} iff its set of histories is
  opaque.
\end{definition}
In \refdef{def:opaque}, conditions $[he] \equiv hs$ and
$\prec_{[he]} \subseteq \prec_{hs}$ establish the ordering constraints
and the requirement that $hs \in {\cal S}$ ensures the memory
semantics constraints. Finally, the prefix-closure constraints are
ensured because end-to-end opacity is checked for each prefix of
$[he]$.

\begin{example}
  The history in \refex{ex:1} is opaque; a corresponding sequential
  history is \smallskip

  \hfill $
  \begin{array}[t]{@{}l@{}}
    \bftmbegin_2 \cdot \bftmread_2(x, 0) \cdot
    \bftmbegin_3 \cdot  {}
    \\
    \bftmwrite_3(x,4)  \cdot \bftmcommit_3
  \end{array}
$\hfill\smallskip

\noindent
Note that reordering of $\tmread_2(x,0)$ and $\tmbegin_3$ is allowed
because their corresponding transactions overlap (even though the
operations themselves do not).
\end{example}

\section{Reducing opacity to linearizability}
\label{sec:reducing-opacity}

In this section, we show that it is possible to reduce a proof of
opacity of an implementation to a proof of linearizability against a
coarse-grained abstraction (where each TM operation is atomic). First,
we recap the definition of linearizability
(\refsec{sec:linearizability}), then we present soundness of the
result, and also establish that the method is complete
(\refsec{sec:soundn-compl}).

\subsection{Linearizability}
\label{sec:linearizability}

As with opacity, the formal definition of linearizability is given in
terms of histories (of invocation/response events); for every
concurrent history an equivalent alternating (invocations immediately
followed by the matching response) history must exist that preserves
real-time order of operations.  The {\em real-time order} on operation
calls\footnote{Note: this is different from the real-time order on
  transactions defined in Section \ref{sec:opacity}} $o_1$ and $o_2$
in a history $h$ is defined as $o_1 \pprec_h o_2$ if the response of
$o_1$ precedes the invocation of $o_2$ in $h$.

Linearizability differs from opacity in that it does not
deal with transactions; thus transactions may still be interleaved in
a matched alternating history.
As with opacity, the given concurrent history may be incomplete, and
hence, may need to be extended using $extend$ and all remaining
pending invocations may need to be removed. We say $lin(h, ha)$ holds
iff both $\complete{h} \equiv ha$ and
$\pprec_{\complete{h}} \subseteq \pprec_{ha}$ hold.
\begin{definition}[Linearizability]
  \label{def:linearized}
  A history $h$ is {\em linearized} by alternating history $ha$ iff
  there exists a history $he \in extend(h)$ such that $lin(he, ha)$.
  A concurrent object is linearizable with respect to a specification
  ${\cal A}$ (a set of alternating histories) if for each concurrent
  history $h$, there is a history $ha \in {\cal A}$ that linearizes
  $h$.
\end{definition}




\subsection{Soundness and completeness}
\label{sec:soundn-compl}

With linearizability formalised, we now present the two main theorems
for our proof method. The first establishes soundness, i.e., states
that one can prove opacity of an implementation by first linearizing
the concurrent operations, then establishing opacity of the linearized
history.
\begin{lemma}[Soundness per history \cite{DDSTW15}]
  \label{lem:opa-run}
  Suppose $h$ is a concrete history. For any alternating history $ha$
  that linearizes $h$, if $ha$ is opaque then $h$ is also opaque.
\end{lemma}
The following theorem lifts this existing result to sets of histories
(of an implementation).
\begin{theorem}[Soundness]
  \label{thm:opa-run}
  Suppose $\mathcal{A}$ is a set of alternating opaque histories. Then
  a set of histories $\mathcal{H}$ is opaque if for each
  $h \in \mathcal{H}$, there exists a history $ha \in \mathcal{A}$ and
  an $he \in extend(h)$ such that $lin(he, ha)$.
\end{theorem}

Next we establish completeness of our proof method, i.e., we show that
if an implementation is opaque, there is some specification $S$ satisfying opacity,
and such that the implementation history is linearizable to $S$.


\begin{lemma}[Existence of Linearization]
  \label{lem:opa-run-complete}
  If $h$ is an opaque history then there exists an alternating history
  $ha$ such that $lin(h, ha)$ and $ha$ is end-to-end opaque.
\end{lemma}
\begin{proof}
  From the assumption that $h$ is opaque, there exists an extension
  $he \in extend(h)$ and a sequential history $hs \in \mathcal{S}$
  such that 
  $\complete{he} \equiv hs$ and
  $\prec_{\complete{he}} \subseteq
  \prec_{hs}$. 
  Our proof proceeds by transposing operations in $hs$ to obtain
  an alternating history $ha$ such that $lin(he, ha)$. Our transpositions
  preserve end-to-end opacity, so $ha$ is end-to-end opaque.

  We consider pairs of operations $o_t$ and $o_{t'}$ such that
  $o_{t}\pprec_{hs} o_{t'}$, but $o_{t'}\pprec_{\complete{he}}o_{t}$,
  which we call {\em mis-ordered pairs}.  If there are no mis-ordered
  pairs, then $lin(he, hs)$, and we are done. Let $o_{t}$ and $o_{t'}$
  be the mis-ordered pair such that the distance between $o_{t}$ and
  $o_{t'}$ in $hs$ is least among all mis-ordered pairs.  Now, $hs$
  has the form $\ldots o_{t} g o_{t'} \ldots$. Note that $g$ does not
  contain any operations of transaction $t$, since if there were some
  operation $o$ of $t$ in $g$, then because opacity preserves program
  order and $o_t\pprec_{hs}o$, we would have
  $o_t\pprec_{\complete{he}}o$. Thus $o,o_{t'}$ would form a
  mis-ordered pair of lower distance, contrary to hypothesis. For a
  similar reason, $g$ does not contain any operations of $t'$. Thus,
  so long as we do not create a new edge in the opacity order
  $\prec_{hs}$, we can reorder $hs$ to (1)
  $\ldots g o_{t'} o_{t}\ldots$ or (2) $\ldots o_{t'} o_{t} g\ldots$
  while preserving opacity. A new edge can be created only by
  reordering a pair of begin and commit operations so that the commit
  precedes the begin.  If $o_t$ is not a begin operation, then we
  choose option (1). Otherwise, note that $o_{t'}$ cannot be a commit,
  because since $o_{t'}\pprec_{\complete{he}} o_t$, $t' \prec t$, and
  thus $t$ could not have been serialised before $t'$.  Since $o_{t'}$
  is not a commit, we can choose option (2).  Finally, we show that
  the new history has no new mis-ordered pairs. Assume we took option
  (1). Then if there is some $o$ in $g$ such that
  $o_t \pprec_{\complete{he}} o$ we would have
  $o_{t'}\pprec_{\complete{he}} o$, and thus $o, o_{t'}$ would form a
  narrower mis-ordered pair. The argument for case (2) is symmetric.
  Thus, we can repeat this reordering process and eventually arrive at
  an end-to-end opaque history $ha$ that has no mis-ordered pairs, and
  thus $lin(he, ha)$.
\end{proof}

\begin{theorem}[Completeness]
If ${\cal H}$ is a prefix-closed set of opaque histories, then there
is some prefix-closed set of opaque alternating histories
${\cal H}'$ such that for each $h\in{\cal H}$ there is some
$h'\in{\cal H}'$ such that $lin(h, ha)$.
\end{theorem}
\begin{proof}
  Let
  ${\cal H}' = \{h'. h' \text{ is opaque and } \exists h\in{\cal
    H}. lin(h, h')\}$.
  Note that both the set of all opaque histories and the set of
  linearizable histories of any prefix-closed set are themselves
  prefix closed.  Thus, ${\cal H}'$ is prefix closed. Because
  ${\cal H}'$ is prefix closed, and each element is end-to-end opaque,
  each element of ${\cal H}'$ is opaque.  For any $h\in{\cal H}$,
  \reflem{lem:opa-run-complete} implies that there is some
  $ha\in{\cal H}'$.
\end{proof}

\section{Proving opacity via linearizability}
\label{sec:proving-opacity-via}

This section describes a proof method for reducing opacity to
linearizability using our running TML example. In
\refsec{sec:coarse-grain-abstr}, we present the coarse-grained
abstraction for TML, which we will refer to as \tmlcga. There are two
distinct proof steps: (1) proving that \tmlcga is opaque; and (2)
proving that TML linearizes to \tmlcga. Both steps could be performed
using any of the existing methods in the literature. We opt for a
fully mechanised simulation-based method for the opacity proof, step
(1) (see \refsec{sec:opac-coarse-grain}); and a model-checking
approach for step (2) (see \refsec{sec:line-against-coarse}). The
combination of the two proofs is much simpler than a single proof of
opacity of the original algorithm.

\subsection{A coarse-grained abstraction}
\label{sec:coarse-grain-abstr}

The coarse-grained abstraction that can be used to prove opacity of
the TML is given in \reflst{alg:tml-cga}. Like TML in
\reflst{alg:tml}, it uses meta-variables $\vloc{t}$ (local to transaction
$t$) and $\vglb$ (shared by all transactions). Each operation is
however, significantly simpler than the TML operations, and performs
the entire operation in a single atomic step.

\begin{algorithm}
  \caption{\tmlcga: Coarse-grained abstraction of TML}
  \label{alg:tml-cga}

  \begin{varwidth}[t]{0.45\textwidth}

    \begin{algorithmic}[1]
      \Procedure{Init}{}
      \State $\vglb$ $\gets$ $0$
      \EndProcedure
      \algstore{tmlcga-ctr}
    \end{algorithmic}

    \begin{algorithmic}[1]
      \algrestore{tmlcga-ctr}
      \Procedure{$\abegin_t$}{}
      \Atomic
      \Await {$\mathit{even}(\vglb)$} \label{tml-cga-begin}
      \State $\vloc{t}$ $\gets$ $\vglb$
      \EndAtomic
      \EndProcedure \smallskip

      \Procedure{$\acommit_t$}{}
      \Atomic
      \If {$\mathit{odd}(\vloc{t})$}
      \State $\vglb$++
      \EndIf
      \EndAtomic
      \EndProcedure
      \algstore{tmlcga-ctr}

    \end{algorithmic}
  \end{varwidth}
  \hfill
  \begin{varwidth}[t]{0.45\textwidth}
    \begin{algorithmic}[1]
      \algrestore{tmlcga-ctr}
      \Procedure{$\aread_t$}{$a$}
      \Atomic
      \If {$\vglb = \vloc{t}$}
      \State \Return $\vmem(a)$
      \Else\  \textbf{abort}
      \EndIf
      \EndAtomic
      \EndProcedure
      \smallskip

      \Procedure{$\awrite_t$}{$a, v$}
      \Atomic
      \If {$\vglb \neq \vloc{t}$}
      \State \textbf{abort}
      \EndIf
      \If {$\mathit{even}(\vloc{t})$}
      \State $\vloc{t}$++; $\vglb$++
      \State $\vmem(a)$ $\gets$ $v$
      \EndIf
      \EndAtomic
      \EndProcedure
    \end{algorithmic}
  \end{varwidth}
\end{algorithm}

\subsection{Opacity of the coarse-grained abstraction}
\label{sec:opac-coarse-grain}

Several methods for proving opacity have been developed, and we are
free to choose any of these methods. We leverage two existing results
from the literature: the TMS2 specification by Doherty \emph{et al.}
\cite{DGLM13}, and the mechanised proof that TMS2 is opaque by Lesani
\emph{et al.}  \cite{LLM12}. Using these results, it is sufficient
that we prove trace refinement (i.e., trace inclusion of visible
behaviour) between \tmlcga and the TMS2 specification. The rigorous
nature of these existing results, means that a mechanised proof of
refinement against TMS2 will also comprise a rigorous proof of opacity
of \tmlcga.

 TMS2 is
formalised using input/output automata \cite{LynchVaan95}, and hence,
our formalisations will also use IOA. Moreover, M{\"u}ller
\cite{MüllerIOA1998} has mechanised the IOA theory (including its
simulation rules) in Isabelle, which is now part of the standard
Isabelle distribution \cite{NipkowPW02}. We thus chose to carry out
our proofs within Isabelle. 
First we define I/O automata.

\begin{definition}
  An \emph{I/O automaton (IOA)} is a labelled transition
  system $A$ with a set of states $\states{A}$, a set of actions
  $acts(A)$ (partitioned into internal and external actions), a set of
  start states $start(A)\subseteq \states{A}$ and a transition
  relation
  $trans(A)\subseteq \states{A}\times acts(A)\times \states{A}$ (so
  that the actions label the transitions).
\end{definition}

Next we formalise refinement and a proof method for it based on
forward simulation. An {\em execution} of an IOA $A$ is a sequence
$\sigma$ of alternating states and actions, beginning with a state in
$start(A)$, such that for all states $\sigma_i$ except the last,
($\sigma_i, \sigma_{i+1}, \sigma_{i+2}) \in trans(A)$. A {\em
  reachable} state of $A$ is a state appearing in an execution of
$A$. An {\em invariant} of $A$ is a predicate satisfied by all
reachable states of $A$. A {\em trace} of $A$ is any sequence of
(external) actions obtained by restricting the actions of $A$ to its
external actions. The set of traces of $A$ represents $A$'s externally
visible behaviour. We say IOA $C$ {\em refines} $A$, denoted
$A \sqsubseteq C$, iff every trace of $C$ is also a trace of $A$. We
say $A$ is \emph{equivalent} to $C$ iff both $A \sqsubseteq C$ and
$C \sqsubseteq A$.

In our setting, each externally visible behaviour consists of the
sequence of invoke and response events, including the input/output
values of reads and writes. Since TMS2 has been specified to only
allow inputs and outputs that result from opaque transactions, a proof
that \tmlcga is a refinement of TMS2 also implies opacity of \tmlcga.

The standard way of verifying a refinement is to use a {\em forward
  simulation} between the implementation and the specification, as
this allows one to verify the refinement in a step-wise manner.  We let
$\Sigma_A^E$ and $\Sigma_A^I$ denote the external and internal actions
of IOA $A$, respectively. Writing $\Trans{C}{cs}{a}{cs'}$ for
$(cs, a, cs')\in trans(C)$, we define the following.
\begin{definition}\label{def:for-sim}
A \emph{forward simulation} from a concrete IOA $C$ to an abstract IOA
$A$ is a relation $R \subseteq \states{C} \times \states{A}$ such that
each of the following holds.  \smallskip

\noindent \emph{Initialisation}.
\begin{gather*}
\forall cs \in start(C) \bullet \exists as \in start(A) \bullet (cs, as) \in R
\end{gather*}
\emph{External step correspondence}.
\begin{gather*}
\begin{array}{@{}l@{}}
  \forall cs\in reach(C), as\in reach(A), a\in \Sigma_C^E, cs' \in
  \states{C} \bullet {} \\
  \quad (cs, as)\in R \land \Trans{C}{cs}{a}{cs'}  \imp\\
  \qquad \exists
  as'\in \states{A} \bullet (cs', as')\in R \land \Trans{A}{as}{a}{as'}
\end{array}
\end{gather*}
\emph{Internal step correspondence}.
\begin{gather*}
\begin{array}[b]{@{}l@{}}
  \forall cs\in reach(C), as\in reach(A), a\in \Sigma_C^I, cs' \in
  \states{C} \bullet {}\\
  \quad (cs, as)\in R \land \Trans{C}{cs}{a}{cs'} \imp  \\
  \qquad (cs', as)\in R \lor{} \\[-3pt]
  \qquad \exists as'\in \states{A}, a' \in \Sigma_A^I
  \bullet (cs', as')\in R \land \Trans{A}{as}{a'}{as'}
\end{array}
\end{gather*}

\end{definition}

\newcommand{\vtime}[1]{\mathit{time}_{#1}}
\newcommand{\vrdset}[1]{\mathit{rdSet}_{#1}}
\newcommand{\vwrset}[1]{\mathit{wrSet}_{#1}}
\newcommand{\vstatus}[1]{\mathit{status}_{#1}}
\newcommand{\vmems}{\mathit{memSeq}}
\newcommand{\latestMem}{\mathit{latestMem}}

\begin{figure}[!t]
\centering
  \begin{minipage}[t]{0.7\columnwidth} {\small
\bgroup
\setlength{\tabcolsep}{1.2em}
\newlength{\myrowsep}
\setlength{\myrowsep}{3em}
\begin{tabular}{@{}l@{~~~~}l@{}}
\action{\beginInv{t}}
{\vstatus{t} = \pcNotStarted}
{\vstatus{t} := \pcBeginPending\\&
 \mathit{beginIdx}_t := \mathit{maxIdx}
                            }
&
\qquad\qquad 
\action{\beginResp{t}}
{\vstatus{t} = \pcBeginPending}
{\vstatus{t} := \pcReady}
\\[1.3\myrowsep]
\action{\readInv{t}{a}}
{\vstatus{t} = \pcReady}
{\vstatus{t} := \pcDoRead(a)}
&\qquad\qquad 
\action{\readResp{t}{v}}
{\vstatus{t} = \pcReadResp(v)}
{\vstatus{t} := \pcReady}
\\[\myrowsep]
\action{\writeInv{t}{a, v}}
{\vstatus{t} = \pcReady}
{\vstatus{t} := \pcDoWrite(a, v)}
&\qquad\qquad 
\action{\writeResp{t}}
{\vstatus{t} = \pcWriteResp}
{\vstatus{t} := \pcReady}
\\[\myrowsep]
\action{\commitInv{t}}
{\vstatus{t} = \pcReady}
{\vstatus{t} := \pcDoCommit}
&\qquad\qquad 
\action{\commitResp{t}}
{\vstatus{t} = \pcCommitResp}
{\vstatus{t} := \pcCommitted}
\\[\myrowsep]
\multicolumn{2}{@{}l}{\action{\abortResp{t}}
{\vstatus{t} \notin \{\pcNotStarted, \pcReady, \pcCommitResp, \pcCommitted, \pcAborted\}}
{\vstatus{t} := \pcAborted}}
\\[\myrowsep]
\action{\doCommitReadOnly{t}{n}}
{\vstatus{t} = \pcDoCommit\\&
 \dom(\vwrset{t}) = \emptyset\\&
 \vidx{t}{n}}
{\vstatus{t} := \pcCommitResp}
&\qquad\qquad 
\action{\doCommitWriter{t}}
{\vstatus{t} = \pcDoCommit\\&
 \vrdset{t} \subseteq \latestMem 
                       }
{\vstatus{t} := \pcCommitResp\\&
\vmems := \\&
\quad \vmems \oplus \mathit{newMem_t}}
\\[2\myrowsep]
\action{\doRead{t}{a, n}}
  {\vstatus{t} = \pcDoRead(a)\\& a\in \dom(\vwrset{t}) \vee {} \\
  & \ \ \vidx{t}{n}}
{\sif\ a \in \dom(\vwrset{t})\ \sthen\\
& \ \ \ \vstatus{t} := \\
& \ \ \ \ \ \ \pcReadResp(\vwrset{t}(a))\\&
 \selse\\&
 \ \ \ v := \vmems(n)(a)\\&
 \ \ \ \vstatus{t} := \pcReadResp(v)\\&
 \ \ \ \vrdset{t} := \\&
 \ \ \ \ \ \vrdset{t}\oplus \{a \mapsto v\}}
&\qquad\qquad 
\action{\doWrite{t}{a, v}}
{\vstatus{t} = \pcDoWrite(a, v)}
{\vstatus{t} := \pcWriteResp\\&
 \vwrset{t} := \\&
 \ \ \vwrset{t}\oplus \{a \mapsto v\}}
\end{tabular}
\egroup\smallskip

{\bf where}
$
\begin{array}[t]{@{}r@{~~}c@{~~}l@{}}
  \mathit{maxIdx} &\sdef & max(\dom(\vmems))%
  \\[1mm]
  \latestMem &\sdef & \vmems(\mathit{maxIdx})\\[1mm]
  \mathit{newMem_t} &\sdef & \{ \mathit{maxIdx}+1 \mapsto (\latestMem\oplus \vwrset{t})\} 
   \\[1mm]
  \vidx{t}{n} &\sdef & \mathit{beginIdx}_t \leq n \leq \mathit{maxIdx} \wedge {}\\
                        & & \vrdset{t} \subseteq \vmems(n) 
\end{array}
$ }
\end{minipage}

\caption{The transitions of TMS2\label{fig:tms2}}
\end{figure}

\noindent
\paragraph{The TMS2 specification.}  The TMS2 specification
\cite{DGLM13} is designed to capture the structural patterns common to
most TM implementations. The actions of TMS2 are given in
\reffig{fig:tms2}. We use notation `$\oplus$' to denote functional
override. For each transition, the first line gives the action name.
--- names of the form $inv_t(\texttt{Op})$ and $resp_t(\texttt{Op})$
are external invocation and response actions of the operation {\tt
  Op}, by transaction $t$ respectively; all others names denote
internal actions.  The transition is \emph{enabled} if all its
preconditions, given after the keyword \textsf{Pre}, hold in the
current state.  State modifications (effects) of a transition are
given as assignments after the keyword \textsf{Eff}.


\OMIT{The status is `$\pcReady$' between reads and writes, and
  `$\pcCommitted$' after the end of the transaction (i.e., when it has
  committed).  Since operations of different transactions may execute
  concurrently, the abstract specification splits executing an
  external operation into several steps, including an {\em invocation}
  and a {\em response}. For example, for \tmread, the external step
  $\readInv{t}{a}$ represents the invocation when reading from
  address $a$, and $\readResp{t}{v}$ represents a read returning
  with value $v$.  In between, an STM implementing TMS2 must at some
  time determine the value it reads. In TMS2 this is represented by
  the internal step $\doRead{t}{a, n}$, which computes $v$ by
  setting $\vstatus{t}$ to $readResp(v)$. The internal actions of TMS2
  (those prefixed by {\tt Do}) correspond to the points at which
  operations ``take effect''.}

Like opacity, TMS2 guarantees that transactions satisfy two critical
requirements: ({\it R1}) all reads and writes of a transaction work
with a \emph{single consistent memory snapshot} that is the result of
all previously committed transactions, and ({\it R2}) the \emph{real-time
  order} of transactions is preserved.

To ensure ({\it R1}), the state of TMS2 includes $\lseq \vmems(0),
\ldots $ $ \vmems(\mathit{maxIdx}) \rseq$, which is a sequence of all
possible memory snapshots.  Initially the sequence consists of one
element, the initial memory $\vmems(0)$. Committing writer transactions
append a new memory $\mathit{newMem}$ to this sequence
(cf. $\doCommitWriter{t}$), by applying the writes of the transaction
to the last element $\vmems(\mathit{maxIdx})$.  To ensure that the
writes of a transaction are not visible to other transactions before
committing, TMS2 uses a {\em deferred update} semantics: writes are
stored locally in the transaction $t$'s write set $\vwrset{t}$ and
only published to the shared state when the transaction commits. Note
that this does not preclude implementations with eager writes (like
\tmlcga).  However, to ensure opacity, such eager implementations must
guarantee that writes are not observable until after the writing
transaction has committed.

All reads in TMS2 must be consistent (i.e., occur from a single memory
snapshot), therefore each transaction $t$ keeps track of all its reads
from memory in a read set $\vrdset{t}$.  A read of address $a$ by
transaction $t$ checks that either $a$ was previously written by $t$
itself ({\bf then} branch of $\doRead{t}{a}$), or that all values
read so far, including $a$, are from the same memory snapshot $n$,
where $\mathit{beginIdx}_t \leq n\leq \mathit{maxIdx}$ (predicate
$\vidx{t}{n}$ from the precondition, which must hold in the
{\bf else} branch).  In the former case the value of $a$ from
$\vwrset{t}$ is returned, and in the latter the value from $\vmems(n)$
is returned and the read set is updated. The read set of $t$ is also
validated when a transaction commits (cf. $\doCommitReadOnlyKW{t}$ and
$\doCommitWriter{t}$). Note that when committing, a read-only
transaction may read from a memory snapshot older than
$\vmems(\mathit{maxIdx})$, but a writing transaction must ensure that
all reads in its read set are from most recent memory (i.e. %
\emph{latestMem} $\vmems(\mathit{maxIdx})$), since its writes will
update the memory sequence with a new snapshot.

To ensure ({\it R2}), if a transaction $t'$ commits before transaction
$t$ starts, then the memory that $t$ reads from must include the
writes of $t'$. Thus, when starting a transaction
(cf. $\beginInv{t}$), $t$ saves the current last index of the memory
sequence, $\mathit{maxIdx}$, into a local variable $\mathit{beginIdx}_t$. When
$t$ performs a read, the check $\vidx{t}{n}$ ensures that that the
snapshot $\vmems(n)$ used has $\mathit{beginIdx}_{t} \le n$, which implies
that the writes of $t'$ are included.\smallskip





\begin{theorem}
  \label{thm:tml-opaque}
  \tmlcga is opaque.
\end{theorem}
\begin{proof}
  We prove forward simulation between the IOA representation of
  \tmlcga and the TMS2 specification.  To construct the simulation
  relation between \tmlcga and TMS2, we start by defining a (partial)
  \emph{step-correspondence} function $\mathit{sc}$, mapping the
  internal actions of \tmlcga to the internal actions of TMS2. We let
  $\undef$ denote an undefined value.
  \begin{align*}
    \mathit{sc}(\aread_t(a))
    &=
      \begin{array}[t]{@{}l@{}}
        \text{if }\vloc{t} = \vglb \text{ then }  \doRead{t}{a,
        \lfloor \vloc{t}/2 \rfloor} \\
        \text{else } \undef
      \end{array}\\
    \mathit{sc}(\awrite_t(a, v))
    &=
      \begin{array}[t]{@{}l@{}}
        \text{if }\vloc{t} = \vglb \text{ then } \doWrite{t}{a, v} \\
        \text{else } \undef
      \end{array}\\
    \mathit{sc}(\acommit_t)
    &=
      \begin{array}[t]{@{}l@{}}
        \text{if } even(\vloc{t}) \text{ then } \\
        \qquad \doCommitReadOnly{t}{\lfloor \vloc{t}/2 \rfloor}\\
        \text{else }  \doCommitWriter{t}
      \end{array}
  \end{align*}
  In addition, we define a global relation between the concrete (i.e.,
  \tmlcga) and abstract (i.e., TMS2) states. This relation states that
  the concrete store, $\vmem$, is equal to the latest store in TMS2,
  with the all the concrete transaction's writes applied to it. This
  is necessary, because TML is an eager algorithm where writes are
  applied immediately, while TMS2 has an write set wherein writes are
  cached. Furthermore, there is a per-transaction simulation relation
  which states that for each transaction $t$, each of the following
  holds. These relate variables $\vloc{t}$ and $\vglb$ from \tmlcga to
  variables of TMS2.
  \begin{itemize}
  \item $\vloc{t}$ is even iff in the write set of TMS2, $\vwrset{t}$, is
    empty.
  \item If $\vloc{t}$ is odd, then in TMS2, $t$ is the currently
    active writer. Furthermore, $\vglb$ is equal to $\vloc{t}$.
  \item The maximum index of $\vmems$ (the sequence of TMS2 stores) is
    less than or equal to $\lfloor\vloc{t}/2\rfloor$.
  \item Each read of $t$ (from the \tmlcga store $mem$) is consistent
    from with the store in $\vmems$ at position
    $\lfloor \vloc{t}/2 \rfloor$.
  \end{itemize}

  The step correspondence function, global relation, and per
  transaction relation are combined into the overall simulation
  relation. The proof is fully mechanised in Isabelle \cite{Web}.
\end{proof}

\subsection{Linearizability against coarse-grained abstraction}
\label{sec:line-against-coarse}

Having established opacity of \tmlcga, we can now focus on
establishing linearizability between TML and \tmlcga, which by
\refthm{thm:opa-run} will ensure opacity of TML. We are free to use
any of the available methods from the literature to prove
linearizability \cite{DongolD15}.  We opt for a model-checking
approach (as opposed to full verification), which provides assurances
of linearizability for finite models. Part of our motivation is to
show that model checking indeed becomes a feasible technique for
verifying opacity, leveraging one of the many methods that have been
developed over the years
\cite{Liu0L0ZD13,CernyRZCA10,BurckhardtDMT10,VechevYY09}. It is also
possible to use static analysis tools (e.g., \cite{Vafeiadis10}) or to
perform full verification (e.g., \cite{SchellhornDW14}).

We use the PAT model checker \cite{SunLDP09}, which enables one to
verify trace refinement (in a manner that guarantees linearizability)
without having to explicitly define invariants, refinement relations,
or linearization points of the algorithm. Interestingly, the model
checker additionally shows that, for the bounds tested, TML is
\emph{equivalent} to \tmlcga, i.e., both produce exactly the same set
of observable traces (see \reflem{lem:tml-eq-cga} below).

\begin{figure}[t]
$\begin{array}{l}
\abegin(t) = inv.\tmbegin.t \rightarrow \abegin\texttt{Loop}(t);\\[0.75em]
\abegin\texttt{Loop}(t) =\\
\quad \mathtt{ifa} (\mathit{even}(\vglb)) \{\\
\quad\quad \mathtt{tau}\{\vloc{}[t] = \vglb;\} \rightarrow ret.{\tmbegin}.t \rightarrow \mathtt{Skip}\\
\quad \}\ \mathtt{else}\ \{\\
\quad\quad \mathtt{tau} \rightarrow \abegin\texttt{Loop}(t)\\
\quad \};
\end{array}$
\caption{\tmlcga $\abegin_t$ procedure in PAT}
\label{fig:pat}
\end{figure}

PAT allows one to specify algorithms using a CSP-style syntax
\cite{Hoare78}. However, in contrast to conventional CSP, events in
PAT (including $\tau$ events) are arbitrary programs assumed to
execute atomically---as such they can directly modify shared state,
and do not have to communicate via channels with input/output events
as in other CSP based model checkers (c.f., FDR3~\cite{fdr3}). This enables
our 
transactional memory algorithms to be implemented very naturally in
PAT. As an example, we give the PAT encoding of the $\abegin_t$
operation from \tmlcga in \reffig{fig:pat}, where $inv.\tmbegin.t$ and
$ret.{\tmbegin}.t$ are observable events corresponding to invoking and
returning from a begin operation. Internal actions are specified using
the $\mathtt{tau}$ keyword, and $\mathtt{ifa}$ is a built-in keyword
that evaluates the test and executes the next event as a single atomic
step. Thus, $\abegin(t)$ fires an (observable) $inv.\tmbegin.t$ event
then executes as $\abegin\texttt{Loop}(t)$. If $glb$ is even, it fires
the observable $ret.{\tmbegin}.t$ event and terminates by executing
$\mathtt{Skip}$, otherwise it retries $\abegin\texttt{Loop}(t)$.

Note that interleaving may occur between the external events
$inv.\tmbegin.t$ and $ret.\tmbegin.t$. However, because the main
action occurs as a single atomic step, this encoding corresponds to
Lynch's \emph{canonical automata} \cite{Lynch96}, which are guaranteed
to be linearizable with respect to the atomic \tmlcga where the
invocation, main action and atomic object are executed as a single
atomic step. Such canonical encodings are also used by Doherty et
al. for verifying concurrent data structures \cite{DGLM04} and is the
prescribed method of proving linearizability in PAT~\cite{Liu0L0ZD13}.

The overall behaviour of the algorithm can be
specified as the interleaving of $N$ transactions where each
transaction begins and then does some number of reads and writes
before either committing successfully or aborting.
\begin{align*}
&\mathtt{Transaction}(t) = \abegin(t);\mathtt{ATXReadWrites}(t);\\
&\mathtt{TML\_CGA}() = \mathtt{|||}t:{0..N-1}@\mathtt{Transaction}(t);
\end{align*}

Once both the coarse-grained and fine-grained algorithms have been
implemented within PAT, trace equivalence (and thus linearizability)
can be checked automatically via refinement using:
\begin{align*}
&\texttt{\#assert TML$()$ refines TML-CGA$();$}\\
&\texttt{\#assert TML-CGA$()$ refines TML$();$}
\end{align*}
Obviously, this does not give us a full proof that TML linearizes to
TML-CGA, as model-checking of course only checks up to a certain
number of transactions with a limited amount of memory.

We thus obtain the following lemma, where constant $\mathit{SIZE}$
denotes the size of the memory (i.e., number of addresses) and
constant $V$ for the possible values in these addresses. The proof is
via model checking using PAT \cite{Web}.

\begin{lemma}
  \label{lem:tml-eq-cga}
  For bounds $N = 3$, $\mathit{SIZE} = 4$, and $V = \{0,1,2,3\}$, as
  well as $N = 4$, $\mathit{SIZE} = 2$, and $V = \{0,1\}$, \tml is
  equivalent to \tmlcga.
\end{lemma}

\section{The NORec algorithm}
\label{sec:norec-algorithm}

\begin{algorithm}[t]
  \caption{NORec pseudocode}
  \label{alg:norec}

\begin{varwidth}[t]{0.45\textwidth}
\begin{algorithmic}[1]
  \Procedure{$\tmbegin_t$}{}
  \State {\bf do}\ $\vloc{t}$ $\gets$ $\vglb$
  \State {\bf until} {$\mathit{even}(\vloc{t})$} 
  \EndProcedure\medskip

  \Procedure{$\Validate_t$}{}
  \While{true}
  \State $\vtime{t}$ $\gets$ $\vglb$\label{nr-val-restart}
  \If {$\mathit{odd}(\vtime{t})$} 
   \State
  {\bf goto} \ref{nr-val-restart}
  \EndIf
  \For {$a \mapsto v \in \vrdset{t}$} \If {$mem(a) \neq v$} \State \textbf{abort}
  \EndIf
  \EndFor
  \If {$\vtime{t} = \vglb$}
  \State \Return $\vtime{t}$
  \EndIf
  \EndWhile
  \EndProcedure
  \algstore{norec-ctr}
\end{algorithmic}
\end{varwidth}
\hfill
\begin{varwidth}[t]{0.45\textwidth}
  \begin{algorithmic}[1]
    \setcounter{ALG@line}{20}

    \Procedure{$\tmwrite_t$}{$a, v$}
    \State $\vwrset{t}$ $\gets$
    \Statex \qquad $\vwrset{t} \oplus \{a \mapsto v\}$
    \EndProcedure \medskip

    \Procedure{$\tmread_t$}{$a$}
    \If {$a \in \mathit{dom}(\vwrset{t})$}
    \State \Return $\vwrset{t}(a)$
    \EndIf
    \State $v_t$ $\gets$ $\vmem(a)$
    \While {$\vloc{t}$ $\neq$ $\vglb$}
    \State $\vloc{t}$ $\gets$ $\Validate_t$
    \State $v_t$ $\gets$ $\vmem(a)$
    \State $\vrdset{t}$ $\gets$
    \Statex \quad\qquad $\vrdset{t} \oplus \{a \mapsto v_t\}$
    \State \Return $v_t$
    \EndWhile
    \EndProcedure
\end{algorithmic}
\end{varwidth}
  \begin{algorithmic}[1]
    \medskip

    \algrestore{norec-ctr}

    \Procedure{$\tmcommit_t$}{}
    \If {$\vwrset{t}$ = $\emptyset$}  \Return
    \EndIf
    \While {$!\mathit{cas}(\vglb, \vloc{t}, \vloc{t} + 1)$}
    \State $\vloc{t}$ $\gets$ $\Validate_t$
    \EndWhile
    \For {$a \mapsto v \in \vwrset{t}$}
    \State $\vmem(a)$ $\gets$ $v$
    \EndFor
    \State $\vglb$ $\gets$ $\vloc{t} + 2$
    \EndProcedure
  \end{algorithmic}
\end{algorithm}

In this section, we show that the method scales to more complex
algorithms. In particular, we verify the NORec algorithm by
Dalessandro \emph{et al.} \cite{DBLP:conf/ppopp/DalessandroSS10} (see
\reflst{alg:norec}), which is one of the best performing STMs that
provides both privatisation and publication safety.

The proof steps for NORec proceeds as with TML. Namely, we construct a
coarse-grained abstraction, \noreccga (see \reflst{alg:norec-cga}),
verify that \noreccga is opaque, then show that NORec linearizes to
\noreccga. As with TML, we do not perform a full verification of
linearizability, but rather, model check the linearizability part of
the proof using PAT.

\newcommand{\tmsv}{TMS3\xspace}

The proof that \noreccga is opaque proceeds via forward simulation
against a variant of TMS2 (\tmsv), which does not require read-only
transactions to validate during their commit, matching the behaviour
of \norec more closely. In particular, \tmsv is identical to TMS2
except that $\doCommitReadOnly{t}{n}$ in the TMS2 transition relation
is replaced by
\begin{align*}
  \action{\doCommitReadOnlyKW{t}}
  {\vstatus{t} = \pcDoCommit \land
  \dom(\vwrset{t}) = \emptyset}
  {\vstatus{t} := \pcCommitResp}
\end{align*}
where $\vidx{t}{n}$ is no longer required in the
precondition. Making this change to TMS2 greatly simplifies the
simulation relation for \noreccga. We have verified within Isabelle
that \tmsv is equivalent to the standard definition of TMS2 from
\reffig{fig:tms2}.

\begin{theorem}
  \tmsv is equivalent to TMS2.
\end{theorem}
\begin{proof}
  We first prove that \tmsv refines TMS2 (soundness) by forward
  simulation. The proof is straightforward, as the simulation relation
  between the two automata is the equality relation, except in the
  case where we must show that $\doCommitReadOnlyKW{t}$ in \tmsv\
  simulates $\doCommitReadOnly{t}{n}$ in TMS2. In this case we must
  provide a valid index $n$ to show that \tmsv\ simulates TMS2. In
  order to accomplish this, we prove an additional invariant to \tmsv\
  which states that there is always some valid index in the stores
  list for any \emph{in-flight transaction}, i.e., a transaction that
  has completed its begin operation. Of course, this invariant would
  also hold for TMS2 itself. Next, we show that TMS2 refines \tmsv
  (completeness), which also proceeds via forward simulation. Once
  again, we use equality as the forward simulation relation, and no
  new invariants need to be introduced. As with our other proofs, the
  Isabelle code may be downloaded~\cite{Web}.
\end{proof}

We are now ready to prove opacity of \noreccga.
\begin{theorem}
  \noreccga is opaque.
\end{theorem}
\begin{proof}

  As was the case for \tmlcga (see \refthm{thm:tml-opaque}), there is
  a simple correspondence between the internal actions of \noreccga
  and the internal actions of \tmsv given by a (partial) step
  correspondence function:
  \begin{align*}
    \mathit{sc}(\aread_t(a))
    &= \begin{array}[t]{@{}l@{}}
         \text{if } a \!\in\! \mathit{dom}(\mathit{wrSet_t}) \lor
         rdSet_t \subseteq mem \\
         \text{then } \doRead{t}{a, N}
         \\
         \text{else } \undef
       \end{array} \\
    \mathit{sc}(\awrite_t(a, v)) &= \doWrite{t}{a, v} \\
    \mathit{sc}(\acommit_t)
    &=
      \begin{array}[t]{@{}l@{}}
        \text{if } \mathit{wrSet}_t = \emptyset \text{ then }
        \doCommitReadOnlyKW{t} \\
        \text{else } \doCommitWriter{t}
      \end{array}
  \end{align*}
  where $N$ is the number of previous transactions that have
  successfully committed. This is implemented as an auxiliary variable
  in the concrete CGA that is incremented whenever it performs the
  \acommit\ action. The rest of the simulation relation is
  considerably simpler than for \tmlcga---we only require that the
  read and write set of the concrete and the abstract states are the
  same for each transaction, that the number of commits $N$ is equal
  to the maximum index in the \tmsv stores list and that the latest
  \tmsv store is equal to the concrete \noreccga store. With the
  simulation relation given by $\mathit{sc}$ and these properties, the
  simulation proof against \tmsv is relatively straightforward.
\end{proof}

\tmsv and \noreccga are quite similar in many respects. They both use
read and write sets in the same way, and write-back lazily during the
commit. The only additional information needed in the simulation is
keeping track of the number of successful commits in \noreccga. Thus,
the simulation relation and refinement proof is easier than the proof
between \tmlcga and TMS2.

\begin{algorithm}[t]
  \caption{\noreccga: Coarse-grained abstraction of NORec}
  \label{alg:norec-cga}

\begin{varwidth}[t]{0.6\textwidth}
\begin{algorithmic}[1]
  \Procedure{$\abegin_t$}{}
  \State \Return
  \EndProcedure\medskip

  \Procedure{\acommit}{$t$}
  \Atomic
  \If {$\vwrset{t} = \emptyset$}  \Return
  \ElsIf {$\vrdset{t} \subseteq \vmem$} $\vmem$ $\gets$ $\vmem \oplus \vwrset{t}$
  \Else\ \textbf{abort}
  \EndIf
  \EndAtomic
  \EndProcedure\medskip

  \Procedure{$\awrite_t$}{$a,v$}
  \State $\vwrset{t}$ $\gets$ $\vwrset{t} \oplus \{a \mapsto v\}$
  \EndProcedure\medskip

  \Procedure{$\aread_t$}{$a$}
  \Atomic
  \If {$a \in \mathit{dom}(\vwrset{t})$} \Return $\vwrset{t}$
  \ElsIf {$\vrdset{t} \subseteq \vmem$} \Return $\vmem(a)$
  \Else\ \textbf{abort}
  \EndIf
  \EndAtomic
  \EndProcedure
\end{algorithmic}
\end{varwidth}
\end{algorithm}

Next, we have the following lemma, which is proved via model checking
using PAT \cite{Web}.
\begin{lemma}
  \label{lem:norec-eq-cga}
  For bounds $N = 2$, $\mathit{SIZE} = 2$ and $V = \{0,1\}$,
  \norec is equivalent to \noreccga.
\end{lemma}

Proving opacity directly, i.e., by showing \norec directly implements
\tmsv would be much trickier as we would need to concern ourselves
with the locking mechanism it employs during the commit to guarantee
that the write-back occurs atomically. However, this locking mechanism
is effectively only being used to guarantee linearizability of the
\norec commit operation, so it need not occur in the opacity proof.

Lesani \emph{et al.} directly verified opacity of NORec \cite{Lesani2012} via
simulation against the TMS2 specification. In comparison to our
approach, Lesani \emph{et al.} introduce several layers of intermediate
automata, performing the full simulation proof in a step-wise
manner. Each of layer adds additional complexity and design elements
of the NORec algorithm to the abstract TMS2 specification. A full
(in-depth) comparison of against this existing proof \cite{Lesani2012}
has however not been possible because these details are not publicly
available. In contrast, we have developed a simple coarse-grained
abstraction \noreccga that generates opaque traces, reducing the proof
of opacity for \norec to linearizability against \noreccga.

\section{Relaxed memory}
\label{sec:relaxed-memory}
We now demonstrate that our method naturally extends to reasoning
about opacity of TM implementations under relaxed memory. We will
focus on TSO in this Section, but our arguments and methods could be
extended to other memory models. Note that we cannot employ a data-race
freedom argument \cite{Adve:DRF} to show that TML or NOrec running on TSO are equivalent
to sequentially consistent versions of the algorithms. This is because
transactional reads can race with the writes of committing transactions
(this is true even when we consider the weaker {\em triangular-race freedom}
condition of \cite{triang-race-freedom}).
This racy behaviour is typical for software transactional memory
implementations.

There are two possibilities for verifying our TM algorithms on TSO. (1)
Leveraging a proof of opacity of the implementation under sequential
consistency then showing that the relaxed memory implementation
refines this sequentially consistent implementation. (2) Showing that
the implementation under relaxed memory linearizes to the
coarse-grained abstraction directly. This approach simply treats an
implementation executing under a particular memory model as an
alternative implementation of the CGA algorithm in question.

In this paper, we follow the second approach, which shows that model
checking linearizability of TSO implementations against a
coarse-grained abstraction is indeed feasible. We verify both TML and
\norec under TSO within the PAT model checker. To do this we encode a
buffer for each transaction where writes to main memory are cached, a
separate $\mathtt{Flusher}()$ process placed in parallel with the
transactions non-deterministically flushes buffers at random, as
shown:
\begin{align*}
&\mathtt{Flusher}() = []i:{0..(N-1)}@(\mathtt{Flush}(i);(\mathtt{Skip}[]\mathtt{Flusher}()));\\
&\mathtt{TML}() = (|||t:{0..N-1}@\mathtt{Transaction}(t))|||\mathtt{Flusher}();
\end{align*}
Writes by one transaction will not become visible to another
transaction until they are flushed. Due to the significantly increased
state-space of the model with these added buffers, checking the
relaxed-memory versions of the algorithms takes significantly longer than
the ordinary fine-grained algorithms. 

Due to the transitivity of trace inclusion, the proof proceeds by
showing that the concrete implementation that executes using relaxed
memory semantics linearizes to its corresponding coarse-grained
abstraction.

We use constant $\mathit{BUFSIZE}$ to bound the maximum size of the
local buffer for each transaction.
\begin{lemma}
  \label{lem:tso-Eq-cga}
  For bounds $N = 2$, $\mathit{SIZE} = 2$, $\mathit{BUFSIZE} = 2$ and
  $V = \{0,1\}$, \tml under TSO is equivalent to \tmlcga and \norec
  under TSO is equivalent to \noreccga.
\end{lemma}

\section{TM design}
\label{sec:tm-design}

Our method of coarse-grained abstraction can potentially contribute to
TM development and allows differences in design to be distinguished at
a higher level of abstraction. Examining the coarse-grained NORec
algorithm, we see that we can avoid validation in the \aread\ and
\tmread\ operations if we have already previously read the value and
stored it in the read set --- here we can simply return the value from
the read set. This variant in CGA is shown in
Listing~\ref{alg:norec2-cga}. Under the assumption that validation is
a more expensive operation than checking membership of the read set
this adds an additional fast path for repeated reads, potentially
increasing performance. In this section we show that we were able to
quickly verify this variant of the algorithm. Furthermore, we
implemented this variant of NORec in the STAMP benchmark
suite~\cite{stamp} to compare its performance with the standard NORec
algorithm.

\begin{algorithm}[t]
   \caption{\ronorec: Coarse-grained \aread operation and
     implementation \tmread}
   \label{alg:norec2-cga}

\begin{algorithmic}[1]



  \Procedure{$\aread_t$}{$a$}
  \Atomic
  \If {$a \in \dom(\vwrset{t})$} \Return $\vwrset{t}(a)$
  \ElsIf {$a \in \dom(\vrdset{t})$} \Return $\vrdset{t}(a)$
  \ElsIf {$\vrdset{t} \subseteq \vmem$} \Return $\vmem(a)$
  \Else\ \textbf{abort}
  \EndIf
  \EndAtomic
  \EndProcedure
\end{algorithmic}\medskip

\begin{algorithmic}[1]
    \Procedure{$\tmread_t$}{$a$}
    \If {$a \in \dom(\vwrset{t})$} \Return $\vwrset{t}(a)$
    \EndIf

    \If {$a \in \dom(\vrdset{t})$} \Return $\vrdset{t}(a)$
    \EndIf

    \State $v_t$ $\gets$ $\vmem(a)$
    \While {$\vloc{t}$ $\neq$ $\vglb$}
    \State $\vloc{t}$ $\gets$ $\Validate_t$
    \State $v_t$ $\gets$ $\vmem(a)$
    \State $\vrdset{t}$ $\gets$ $\vrdset{t} \oplus \{a \mapsto v_t\}$
    \State \Return $v_t$
    \EndWhile
    \EndProcedure
\end{algorithmic}
\end{algorithm}

\begin{theorem}
  \ronoreccga is opaque.
\end{theorem}
\begin{proof}
  The proof for opacity of \ronoreccga is largely the same as for
  \noreccga. The only significant difference is that we need an
  additional auxiliary variable at the concrete level to keep track of
  the number of validating reads that have occurred in the CGA. This
  is due to how for the ordinary \noreccga, the commit of a read-only
  transaction can be re-ordered to it's last read, while for
  \ronoreccga it can be reordered to the last validating read. For
  details we refer the reader to our Isabelle
  implementation~\cite{Web}.
\end{proof}

\begin{lemma}
  For bounds $N=2$, $\mathit{SIZE} = 2$ and $V= \{0,1\}$, \ronorec is
  linearizable with respect to \ronoreccga.
\end{lemma}

We now show that the three coarse-grained abstractions we have
considered in this paper (\reffig{fig:overview}) are pairwise
distinct.
\begin{lemma}
  \label{lem:distinct}
  \tmlcga, \noreccga, \ronoreccga are pairwise distinct.
\end{lemma}
\begin{proof}
  We automatically generate counter-examples using the PAT model
  checker. First we show \tmlcga is distinct from both \noreccga and
  \ronoreccga. History $h_1$ below is allowed by \tmlcga but not
  allowed by \noreccga or \ronoreccga.  \tmlcga only allows one writer
  at any time, and hence, transaction 1 aborts. However, in \noreccga,
  a write operation never aborts because they are cached in a
  transaction-local read set.
  \begin{align*}
    h_1 = {}
    &
      \begin{array}[t]{@{}l@{}}
        \bftmbegin_0 \cdot \lseq \tmwrite_0(x,0) \rseq \cdot
        \bftmbegin_1 \cdot  {} \\
        \lseq
        \tmwrite_1(x,0),\tmabort_1\rseq
      \end{array}      
  \end{align*}
    History $h_2$ below is valid for both \noreccga and \ronoreccga but  invalid for
    \tmlcga, as a \tmlcga transaction must wait for the writing
    transaction to commit before it can begin.
  \begin{align*}
    h_2 = {}
    &
      \begin{array}[t]{@{}l@{}}
        \bftmbegin_0 \cdot \bftmwrite_0(x,0) \cdot
        \bftmbegin_1
      \end{array}
  \end{align*}
    Thus \noreccga and \tmlcga are distinct. Histories $h_3$ and
    $h_4$ below demonstrate that the \noreccga and \ronoreccga are
    distinct.
  \begin{align*}
    h_3 = {}
    &
      \begin{array}[t]{@{}l@{}}
        \bftmbegin_0 \cdot \bftmwrite_0(x,1) \cdot \lseq \tmcommit_0
        \rseq \cdot \bftmbegin_1 \cdot {} \\
        \bftmread_1(x,0) \cdot \lseq \tmread_1(x),\tmabort_1\rseq
      \end{array}
    \\
    h_4 = {}
    &
      \begin{array}[t]{@{}l@{}}
        \bftmbegin_0 \cdot \bftmwrite_0(x,1) \cdot \lseq \tmcommit_0
        \rseq \cdot \bftmbegin_1 \cdot {}\\
        \lseq \tmread_1(x),
        \ret{\tmcommit}_0,\ret{\tmread}_1(0) \rseq \cdot \bftmread_1(x,0)
      \end{array}
  \end{align*}

    In $h_3$, transaction $1$ reads $x$ and then reads $x$ again,
    aborting the second read because another writer (transaction $0$)
    has committed. This history is allowed by \noreccga but not
    \ronoreccga. History $h_4$ is similar to $h_3$, but the second read
    by transaction $0$ succeeds because it reads from local read set
    rather than main memory (which allows it to bypass
    validation). History $h_4$ is allowed by \ronoreccga, but not
    \noreccga. In general, \ronoreccga aborts strictly less often than
    \noreccga as it contains fewer code paths that lead to an abort occurring.
\end{proof}

Finally, we conclude that the three algorithms we have considered are
pairwise distinct, i.e., that the three algorithms considered indeed
implement different TM designs.

\begin{lemma}
  \tml, \norec and \ronorec are pairwise distinct.
\end{lemma}
\begin{proof}
  The proof follows from the equivalence between each algorithm and
  its coarse-grained counterpart (Lemmas~\ref{lem:tml-eq-cga} and
  \ref{lem:norec-eq-cga}) together with \reflem{lem:distinct}.
\end{proof}

\paragraph{Experimental results.}
Having checked opacity for our variant of \norec, we implemented the
algorithm in C using the STAMP benchmarking suite~\cite{stamp} for
transactional memory algorithms, modifying an existing \norec
algorithm by Diegues \emph{et al}~\cite{norecc}. Contrary to our
initial expectations, a naive implementation of \ronorec is actually
slower than ordinary \norec, sometimes up to ten times slower. The
reasons for this are as follows:
\begin{itemize}
\item By avoiding validation for certain reads, we allow writing
  transactions that would otherwise abort to continue running, wherein
  they will invariably (and necessarily) abort in the commit due to
  validation therein. In general, it seems better for a transaction to
  fail fast rather than waste time in a doomed state. Tests using the
  STAMP suite indicate that, for many benchmarks using this variant of
  \norec, there are millions of transactions performing unnecessary
  work before they eventually abort when they commit.
\item In software TM, the difference in performance between validating
  the read set (increasing contention to main memory) compared to
  checking membership of the read set (which is often implemented as a
  list) is not significant enough to warrant the existence of the
  additional fast-path check. 
\end{itemize}
We note however that we have shown opacity for transactions that
return values from their read set whenever possible. Any
implementation of \ronoreccga where we allow re-reading of values from
the read-set only in certain circumstances, e.g., in read-only
transactions, and after the last write in a transaction are still
valid. This avoids the first performance issue above. Such variants
between \norec and \ronorec are clearly opaque. We have benchmarked
both variants, and shown that both perform as well as standard \norec,
yet transactions abort less often, and therefore represent meaningful
(albeit small optimisations) over \norec.

Overall, our method has allowed us to identify an optimisation at a
high level of abstraction and rapidly verify opacity of a variant of
the NORec algorithm. Modifying the opacity proof and model checking
the result took no more than a days worth of work. However, this also
demonstrates the importance of concrete benchmarks to evaluate the
usefulness of such modifications.

\section{Conclusions}

\paragraph{Related work.}
With the widening applicability of TM, there has been an increased
interest in verifying opacity of the underlying algorithms. There are
direct proofs \cite{DDSTW15} as well as proofs by simulation
\cite{Lesani2012} that make use the intermediate specification TMS2
\cite{DGLM13,LLM12}. A detailed comparison with this existing (direct)
simulation-based proof against TMS2 \cite{LLM12} is given in
\refsec{sec:norec-algorithm}. Further comparisons against Derrick
\emph{et al.}'s method \cite{DDSTW15} are give below. Anand \emph{et
  al.} present a proof method for \emph{conflict opacity}, which is a
subclass of opacity \cite{ASS16}, however, these proofs are not
mechanised. Lesani has developed an alternative criteria,
\emph{markability} \cite{LesaniP14,Lesani14} that allows one to prove
opacity by checking that certain conflict ordering properties are
satisfied. The markability technique requires reasoning about two
orders: an {\em effect order} relating transactions, and an {\em
  observation order}, relating writes and reads. Using the technique
involves checking consistency conditions on both these orders, rather
than the simple real-time order of linearizability.  Another
conflict-based technique has also been developed by Guerraoui \emph{et
  al.} \cite{GuerraouiHS10,GHS08,GuerraouiHS11} allowing one to reduce
a proof of opacity to checking opacity of a system with only two
threads and two variables. This reduction depends on the algorithm in
question satisfying a number of properties. 
One point of difference compared with our methods is that Guerraoui
\emph{et al.}'s conflict relations do not check the values within each
address.

\paragraph{Contributions.} Our main contributions for this paper are as
follows.
\begin{enumerate}
\item We have developed a complete method for reducing proofs of
  opacity to the simpler and more widely studied condition
  linearizability. Soundness is proved via an existing result by
  Derrick \emph{et al.} \cite{DDSTW15}. These results bring together the
  previously disconnected worlds of linearizability and opacity
  verification, and allows one to reuse the vast literature on
  linearizability verification \cite{DongolD15}, as well as the
  growing literature on opacity verification (to verify the
  coarse-grained abstractions).
\item We have demonstrated our technique using the TML algorithm and shown
  that the method extends to more complex algorithms by verifying the
  NORec algorithm. As part of this verification, we discover a
  variation of the TMS2 specification, \tmsv that does not require
  validation when read-only transactions commit. We show that TMS2 is
  equivalent to \tmsv.
\item We have shown that the method naturally copes with relaxed memory by
  verifying both TML and NORec are opaque under TSO, and neither
  requires introduction of additional fence instructions. These
  examples show that a single coarse-grained abstraction is valid for
  more than one implementation.
\item We developed a variation of NORec which allows reads to be
  re-read from the read set, and demonstrated that this variation
  aborts less often than the existing NORec algorithm. Consideration
  of opacity at higher levels of abstraction elucidates design
  elements that are not immediately apparent at the level of concrete
  code. We were able to quickly verify and test this modified
  algorithm. Through benchmarking, we developed variations with
  meaningful optimisation, that demonstrate improvements to the
  original implementation in specific circumstances.
\end{enumerate}
Verifying opacity directly is difficult --- because a read-only
transaction may serialise against earlier versions of memory, opacity
requires one to keep track of all stores created by committed writer
transactions. In comparison, linearizability only requires one to keep
track of the latest memory snapshot, which simplifies the proof of an
implementation.  From a verification perspective, a coarse-grained
abstraction is straightforward to construct.

The basic idea of linearizing fine-grained operations to verify
opacity also appears in \cite{DDSTW15}, where the TML algorithm has
been verified. However, their coarse-grained abstraction is very
different from ours. Namely, the operations only modify memory and
does not utilise any synchronisation code. As a result, their proof
must couple an abstraction with the original TML implementation and
introduce an explicit history variable. The proof proceeds via
induction on the histories generated by TML --- these show that each
step of TML preserves opacity. However, it is unclear if such an
inductive method could scale to more complex algorithms such as NORec,
or to include relaxed memory algorithms.

In contrast, our proofs achieve a clear separation of concerns between
opacity (of a coarse-grained abstraction) and linearizability (of an
implementation), i.e., opacity of a coarse-grained abstraction is
verified independently. This separation enables one to distinguish
design elements of opaque algorithms at a higher level of abstraction,
separating design aspects from synchronisation elements in an
implementation that ensure atomicity of fine-grained transactional
operations. These insights have been used to develop a new variation
of NORec with fast-path read-only transactions.

\paragraph{Experiences.} Our experiences suggest that our techniques
do indeed simplify proofs of opacity (and their mechanisation). We
evidence this by verifying several algorithms with relatively little
effort. Our completeness result ensures that the technique is always
applicable.

Opacity of each coarse-grained abstraction is generally trivial to
verify (our proofs are mechanised in Isabelle), leaving one with a
proof of linearizability of an implementation against this
abstraction. The second step is limited only by techniques for
verifying linearizability. We have opted for a model checking approach
using PAT, which enables linearizability to be checked by
automatically checking refinement. These show that the coarse-grained
abstractions we have developed are actually equivalent to their
implementations up to some bound on the number of transactions and
size of the store. It is of course also possible to perform a full
verification of linearizability.

Overall, PAT was well suited for the task of model checking
linearizability. We did not need to provide PAT with any information
on the linearization points, invariants, or simulation relations
between the concrete implementations and CGA abstractions. Once we had
encoded our TM algorithms in PAT, verification was fully
automatic. The task of implementing the TM algorithms in PAT was made
particularly simple due to its support for CSP style processes
combined with shared mutable state. Other CSP style model checkers
such as FDR3~\cite{fdr3} have a more functional specification language
disallowing such mutable state. To implement in our algorithms in such
a system would require encoding the shared state as message passing
between processes, which would significantly increase the distance
between the model and the natural pseudocode implementations of the
algorithms, which in PAT, is fairly minimal. We note that while PAT
did not appear to effectively take advantage of the 24 cores on the
machine we ran our model-checking on, our verifications were limited
by memory usage rather than processor speed. For example, model
checking TML with four transactions used over 40GB of memory. It is
not clear whether this high memory usage could be reduced with another
model checker, or if it is inherent in the complexity of the models
themselves.

\paragraph{Future work.} Our work suggests that to fully verify a TM
algorithm using coarse-grained abstraction, the bottleneck to
verification is the proof of linearizability itself
\cite{DongolD15}. It is therefore worthwhile considering whether
linearizability proofs can be streamlined for transactional
objects. For example, Bouajjani \emph{et al.} have shown that for
particular inductively-defined data structures, linearizability can be
reduced to state reachability \cite{BouajjaniEEH15}. Exploration of
whether such methods apply to transactional objects remains a topic
for future work. Establishing this link would be a useful result ---
it would allow one to further reduce a proof of opacity to a proof of
state reachability.

Verifying linearizability sometimes requires information about the
possible (future) executions of an algorithm
\cite{DongolD15,HeWi90,SchellhornDW14}. The examples we have
considered have not required consideration of (speculative)
future-based behaviour, but such TM algorithms do exist, e.g., TL2
\cite{DBLP:conf/wdag/DiceSS06}. However, for such algorithms, a direct
proof of opacity also requires one to use the same futuristic
information~\cite{Lesani2012}. Development of coarse-grained
abstractions for algorithms such as TL2 remains a topic of future work
--- it is worth noting that our completeness result ensures that this
can be done.

This work can also contribute to the \emph{characterisation} of
algorithms at a high-level of abstraction. That is, by verifying a
number of other algorithms from the literature, it will be possible to
extend the tree in \reffig{fig:overview}, and understand particular
design features at a coarse-grained level. 





We thank John Derrick for helpful discussions and funding from EPSRC
EP/N016661/1.

\bibliographystyle{plain}

\bibliography{references}

\end{document}

%% file: overview.pspdftex
\begin{picture}(0,0)%
\includegraphics{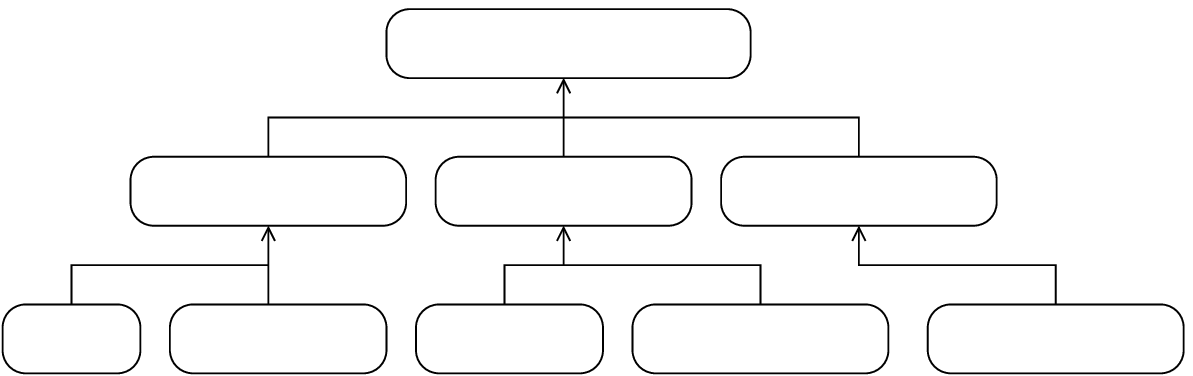}%
\end{picture}%
\setlength{\unitlength}{4144sp}%
\begingroup\makeatletter\ifx\SetFigFont\undefined%
\gdef\SetFigFont#1#2#3#4#5{%
  \reset@font\fontsize{#1}{#2pt}%
  \fontfamily{#3}\fontseries{#4}\fontshape{#5}%
  \selectfont}%
\fi\endgroup%
\begin{picture}(6178,1689)(1879,-2863)
\put(5896,-1501){\makebox(0,0)[lb]{\smash{{\SetFigFont{12}{14.4}{\rmdefault}{\mddefault}{\updefault}{\color[rgb]{0,0,0}Opacity}%
}}}}
\put(3151,-2761){\makebox(0,0)[b]{\smash{{\SetFigFont{12}{14.4}{\rmdefault}{\mddefault}{\updefault}{\color[rgb]{0,0,0}TML-TSO}%
}}}}
\put(2206,-2761){\makebox(0,0)[b]{\smash{{\SetFigFont{12}{14.4}{\rmdefault}{\mddefault}{\updefault}{\color[rgb]{0,0,0}TML}%
}}}}
\put(4456,-2086){\makebox(0,0)[b]{\smash{{\SetFigFont{12}{14.4}{\rmdefault}{\mddefault}{\updefault}{\color[rgb]{0,0,0}\noreccga}%
}}}}
\put(5806,-2086){\makebox(0,0)[b]{\smash{{\SetFigFont{12}{14.4}{\rmdefault}{\mddefault}{\updefault}{\color[rgb]{0,0,0}\ronoreccga}%
}}}}
\put(3106,-2086){\makebox(0,0)[b]{\smash{{\SetFigFont{12}{14.4}{\rmdefault}{\mddefault}{\updefault}{\color[rgb]{0,0,0}\tmlcga}%
}}}}
\put(4231,-2761){\makebox(0,0)[b]{\smash{{\SetFigFont{12}{14.4}{\rmdefault}{\mddefault}{\updefault}{\color[rgb]{0,0,0}NORec}%
}}}}
\put(5356,-2761){\makebox(0,0)[b]{\smash{{\SetFigFont{12}{14.4}{\rmdefault}{\mddefault}{\updefault}{\color[rgb]{0,0,0}NORec-TSO}%
}}}}
\put(6706,-2761){\makebox(0,0)[b]{\smash{{\SetFigFont{12}{14.4}{\rmdefault}{\mddefault}{\updefault}{\color[rgb]{0,0,0}\ronorec}%
}}}}
\put(4456,-1411){\makebox(0,0)[b]{\smash{{\SetFigFont{12}{14.4}{\rmdefault}{\mddefault}{\updefault}{\color[rgb]{0,0,0}TMS2 specification}%
}}}}
\put(6796,-2356){\makebox(0,0)[lb]{\smash{{\SetFigFont{12}{14.4}{\rmdefault}{\mddefault}{\updefault}{\color[rgb]{0,0,0}proof}%
}}}}
\put(6796,-2176){\makebox(0,0)[lb]{\smash{{\SetFigFont{12}{14.4}{\rmdefault}{\mddefault}{\updefault}{\color[rgb]{0,0,0}Linearizability}%
}}}}
\put(5896,-1681){\makebox(0,0)[lb]{\smash{{\SetFigFont{12}{14.4}{\rmdefault}{\mddefault}{\updefault}{\color[rgb]{0,0,0}proof}%
}}}}
\end{picture}%